\documentclass[11pt]{article}
\usepackage{amsmath,amsfonts,amssymb,amsthm}
\usepackage{mathtools}
\usepackage[margin=3cm]{geometry}
\usepackage{tikz}
\usepackage[hidelinks=false]{hyperref}
\usepackage[nameinlink]{cleveref}
\usepackage{dsfont}
\usepackage{color, comment}

\theoremstyle{definition}

\newtheorem{theorem}{Theorem}[section]
\newtheorem{proposition}[theorem]{Proposition}
\newtheorem{lemma}[theorem]{Lemma}
\newtheorem{corollary}[theorem]{Corollary}
\newtheorem{conjecture}[theorem]{Conjecture}

\newtheorem{definition}[theorem]{Definition}
\newtheorem{example}[theorem]{Example}

\newtheorem{remark}[theorem]{Remark}

\newcommand{\F}{\mathds{F}}

\newcommand{\colsp}{\text{colsp}}

\newcommand{\rank}{\text{rk}}
\newcommand{\Tr}{\text{Tr}}
\newcommand{\GL}{\text{GL}}

\newcommand{\Acode}{\mathcal{A}}

\newcommand{\Dcode}{\mathcal{D}}

\newcommand{\code}{\mathcal{C}}
\newcommand{\Mcode}{\mathcal{M}}
\newcommand{\Ncode}{\mathcal{N}}
\newcommand{\Tcode}{\mathcal{T}}
\newcommand{\Vcal}{\mathcal{V}}
\newcommand{\Lcal}{\mathcal{L}}

\newcommand{\MDSFam}{M^{\textnormal{H}}_{n,q}}
\newcommand{\MRDFam}{M^{\textnormal{rk}}_{n\times m,q}}

\newcommand{\MDS}{\text{MDS}}
\newcommand{\MRD}{\text{MRD}}

\newcommand{\maxrk}{\text{maxrk}}

\newcommand{\supp}{\text{supp}}

\newcommand{\dd}{\text{d}}
\newcommand{\Hd}{\text{d}_{\text{H}}}
\newcommand{\rkd}{\text{d}_{\text{rk}}}
\newcommand{\srkd}{\text{d}_{\text{srk}}}

\newcommand{\w}{\text{w}}
\newcommand{\Hw}{\text{w}_{\text{H}}}
\newcommand{\rkw}{\text{w}_{\text{rk}}}
\newcommand{\srkw}{\text{w}_{\text{srk}}}

\title{\textbf{A distance-free approach to generalized weights}}

\usepackage{authblk}

\author{Andrea Di Giusto$^1$, Elisa Gorla$^2$, Alberto Ravagnani$^1$}
\affil{$^1$Technische Universiteit Eindhoven, $^2$Université de Neuchâtel}
\date{}

\begin{document}

\setlength{\parindent}{20pt}

\maketitle

\begin{sloppypar}

\begin{abstract}
We propose a unified theory of generalized weights for linear codes  endowed with an arbitrary distance. Instead of relying on supports or anticodes, the weights of a code are defined via the intersections of the code with a chosen family of spaces, which we call a test family. The choice of test family determines the properties of the corresponding generalized weights and the characteristics of the code that they capture. 
In this general framework, we prove that generalized weights are weakly increasing and that certain subsequences are strictly increasing. We also prove a duality result reminiscent of Wei's Duality Theorem. The corresponding properties of generalized Hamming and rank-metric weights follow from our general results by selecting optimal anticodes as a test family.
For sum-rank metric codes, we propose a test family that results in generalized weights that are closely connected to -- but not always the same as -- the usual generalized weights. This choice allows us to extend the known duality results for generalized sum-rank weights to some sum-rank-metric codes with a nonzero Hamming component.
Finally, we explore a family of generalized weights obtained by intersecting the underlying code with MDS or MRD codes.  
\end{abstract}

\section*{Introduction}\label{sec:introduction}

Generalized Hamming (GH) weights are invariants of codes in the Hamming metric, which first rose to popularity when V. Wei linked them to the cryptographic performance of coset coding schemes in the wire tap channel II security model~\cite{wei1991generalized}.
Since then, GH weights have proven to be an effective means of describing the structure of codes, and their properties have been studied extensively.
A connection to the distance/length profile and the trellis complexity of a linear block code, a fundamental parameter to measure the decoding complexity, was discovered in~\cite{forney1994dimension}. 

Generalized weights of linear codes have also been defined and studied for other metrics, such as the rank metric~\cite{kurihara2015relative,ravagnani2016generalized,martinez2017relative} and the sum-rank metric~\cite{camps2022optimal}.
In the more general framework of codes over rings, the properties of generalized weights have been explored over Galois rings~\cite{ashikhmin1998generalized}, finite chain rings~\cite{horimoto2001generalized}, and Frobenius rings~\cite{liao2022relative}. For convolutional codes, different notions of generalized weights have been proposed, in connection with both the free and the column distances, see e.g.~\cite{RosenthalYork1997,cardell2020generalized,gorla2023generalized,gorla2024generalized}. On a different note, the relation between generalized weights and graded Betti numbers of the monomial ideal associated to a code is discussed in~\cite{johnsen2013hamming} for linear block codes endowed with the Hamming metric and in~\cite{gorla2022generalized} for codes over rings endowed with a well-behaved support function.

These different notions of generalized weights share some features: They are usually increasing sequences of natural numbers that measure the minimum weight of a subcode of a given dimension, where the weight function changes depending on the metric used.
A striking similarity is that, regardless of the metric, it is often possible to prove duality statements for generalized weights.
Such statements allow to determine the generalized weights of the dual of a code $\code$ from the weights of $\code$.
Beyond their mathematical interest, duality results play a central role in the security applications of generalized weights, as the performance of a code is often determined by the weights of its dual code.

In~\cite{forney1994dimension}, D. Forney argued that generalized weights have ‘‘more to do with length and dimension than with distance’’. 
This observation may be taken as a starting point for the approach to generalized weights that we adopt in this paper, where we propose a definition of generalized weights that does not involve the distance.

\paragraph{Our contribution.} We define generalized weights of a linear code $\code$ by considering the intersection of $\code$ with a given family of codes $T$.
The families we consider satisfy some assumptions, listed in~\Cref{def:test_family}, and are called test families. Their elements are called test codes.
The weights $\tau_{T,r}$ are defined as the minimum dimension of a test code~$\Tcode$, whose intersection with $\code$ has at least a given dimension. 
In formul\ae
\begin{equation}\label{defwts}
    \tau_{T,r}(\code)=\min\{\dim(\Tcode)\mid\Tcode\in T,\dim(\Tcode\cap\code)\geq r\},
\end{equation}
where $1\leq r\leq\dim(\code)$.
This approach is inspired by the reformulation of the definition of generalized weights proposed in~\cite{ravagnani2016generalized} for the Hamming and rank metrics.

Working with an arbitrary metric, the choice of a test family allows us to associate a sequence of weights to a linear code.
The weights inherit properties from the test family used in their definition.
Our main result is that generalized weights defined as in~\Cref{{defwts}} satisfy a duality theorem.
Many known families of codes are test families: As a running example, we instantiate our construction with the family of standard optimal anticodes in the Hamming metric. 
Standard optimal anticodes in the rank and sum-rank metrics also provide natural examples of test families.

There are two main definitions of generalized weights for codes in the rank metric: 
One always produces invariants~\cite{ravagnani2016generalized}, while the other has a closer link to security and network coding~\cite{kurihara2015relative,martinez2017relative}.
However, both notions satisfy duality statements. Thanks to the concept of test families and to the related duality result, we can shed new light on why both dualities hold. At the same time, we provide a unified treatment for the different definitions of generalized weights in the rank metric. 

Having one distance-free construction that is consistent with known results in the Hamming and rank metrics, it is natural to apply it to sum-rank metric codes. In this context, we propose a definition of generalized weights that agrees with that of~\cite{camps2022optimal} for codes with a trivial Hamming component, but is slightly different in the general case. This allows us to recover the duality result of~\cite{camps2022optimal} and to establish a new duality result for some codes with a non-trivial Hamming part. 

Finally, we consider generalized weights derived from the intersection with the MDS and MRD codes.
Although these are very natural code families to consider, they are not test families in general. This is consistent with the fact that duality does not hold for the corresponding weights.
We provide explicit counterexamples showing that the weights of a code do not determine those of the dual in this context.
This provides a heuristic justification for why the assumptions in our constructions are minimal: The families of MDS and MRD codes lack the regular containment structure of test families, and this directly affects the duality of the corresponding weights.
The absence of a general duality statement does not mean that duality cannot hold in special cases. For example, we prove a duality result for the weights derived from MDS codes under the assumption that the field size $q$ is large enough.
Using results from the literature, we show that, for a sufficiently large $q$, the family of MDS codes is indeed a test family.
Moreover, it turns out that in the same regime of parameters, our weights are equivalent to the code distances, a family of invariants recently introduced in~\cite{camps2025code}.

\paragraph{Structure of this work.} In~\Cref{sec:preliminaries} we recall some preliminaries and fix the notation for the rest of the paper.
Our definition of generalized weights for arbitrary vector spaces and the corresponding duality results are presented in~\Cref{sec:general_theory}, where we use GH weights as a running example.
In~\Cref{sec:rank} we show how to instantiate the construction to obtain the two different notions of generalized weights in the rank metric, with an emphasis on their duality properties.
We move to the study of sum-rank metric codes in~\Cref{sec:sumrank}, where we show how to extend the existing duality results for generalized weights to some codes with nontrivial Hamming part.
In~\Cref{sec:MDS} we treat the weights derived from MDS/MRD codes.
Over fields of large enough size $q$, we relate generalized weights obtained from MDS codes and code distances.

\section{Preliminaries and notation}\label{sec:preliminaries}

Fix integers $0\leq m\leq n$ and let $[m,n]=\{m,m+1,\ldots,n\}$, $[n]=[1,n]$, and $[n]_m=\{i\in[n]:i=0\mod m\}$.
Let $q$ be a prime power and $\F_q$ be the finite field with $q$ elements. We denote by $V$ a finite dimensional vector space of dimension $\dim(V)=\nu$ over $\F_q$.
Linear subspace containment is denoted by $\leq$.
\begin{definition}
A $[\nu,k]_q$ \textit{linear code} in $V$ is a $k$-dimensional vector subspace $\code\leq V$.
\end{definition}

All codes in this work are linear.
The characteristics and performance of codes are measured and compared using parameters.
Examples of parameters of a code $\code$ include its dimension $\dim(\code)=k$ and the dimension of the ambient space $\dim(V)=\nu$, often called the {\em length} of the code.

If the space $V$ is equipped with a nondegenerate bilinear form $\langle\cdot,\cdot\rangle:V\times V\to\F_q$, then one can define dual codes. The reader should be familiar with the standard inner product in $V=\F_q^n$ used to define dual codes in the Hamming metric, see~\Cref{ex:Hamming_space}.

\begin{definition}
Two vectors $x,y\in V$ are orthogonal if $\langle x,y\rangle=0$. The {\em dual} of a code $\code$ is the set of vectors which are orthogonal to all elements of $\code$, that is, $$\code^\perp=\{y\in V\mid \langle x,y\rangle=0\;\mbox{ for all } x\in\code\}.$$
\end{definition}

The natural framework of coding theory is a metric $\F_q$-vector space $(V,\dd)$, where the distance $\dd$ allows for error correction. 
We are especially interested in three examples: the Hamming and rank metric spaces, which are ubiquitous in coding theory, and the sum-rank metric space, which generalizes the previous two and is of interest, for example, in multishot network coding~\cite{napp2018multi} and space time coding~\cite{shehadeh2021space}.

\begin{example}[Hamming metric space]\label{ex:Hamming_space}
Let $(V,\dd)=(\F_q^n,\Hd)$, where $\Hd$ is the \textit{Hamming distance} and $\langle\cdot,\cdot\rangle$ is the \textit{standard inner product}.
For $x,y\in V$, $x=(x_1,\ldots,x_n)$, $y=(y_1,\ldots,y_n)$
    \begin{equation*}
        \Hd(x,y)=|\{i:x_i\neq y_i\}|\quad\textnormal{and}\quad\ \langle x,y\rangle=\sum_{i=1}^nx_iy_i.
    \end{equation*}
The book~\cite{huffman2010fundamentals} is a standard reference on codes in the Hamming metric.
\end{example}

\begin{example}[Rank-metric space]\label{ex:rank_space}
Let $(V,\dd)=(\F_q^{n\times m},\rkd)$, where $\rkd$ is the \textit{rank distance} and the bilinear form is the \textit{trace product}: 
    \begin{equation*}
        \rkd(X,Y)=\rank(X-Y)\quad\textnormal{and}\quad\langle X,Y\rangle=\Tr(XY^t)
    \end{equation*}
for $X,Y\in V$. Here, $\rank(\cdot)$ and $\Tr(\cdot)$ denote the rank and trace functions, respectively, and $Y^t$ denotes the transpose of $Y$.
See~\cite{gorla2021rank} for an introduction to rank-metric codes.
\end{example}

\begin{example}[Sum-rank metric space]\label{ex:sumrank_space}
    Let $V=\F_q^{n_1\times m_1}\times\ldots\times\F_q^{n_t\times m_t}$, where $1\leq n_i\leq m_i$ for all $i=1,\ldots,t$, and $m_1\geq m_2\geq\ldots\geq m_t$.
    The \textit{sum-rank distance} between two $t$-uples of matrices $X=(X_1,\ldots,X_t)$, $Y=(Y_1,\ldots,Y_t)\in V$ is 
    \begin{equation*}
        \srkd(X,Y)=\sum_{i=1}^t\rkd(X_i,Y_i)=\sum_{i=1}^t\rank(X_i-Y_i).
    \end{equation*}
The bilinear form on $V$ is the \textit{sum-trace product}
    \begin{equation*}
        \langle X,Y\rangle=\sum_{i=1}^t\Tr(X_iY_i^t).
    \end{equation*}
The sum-rank metric is a generalization of both the rank metric and the Hamming metric. In fact, if $t=1$, then $V=\F_q^{n_1\times m_1}$ and $\srkd=\rkd$.
Moreover, if $m_1=1$ then $m_i=n_i=1$ for all $i$, then $V=\F_q^t$ and the sum-rank metric $\srkd$ coincides with the Hamming metric $\Hd$ on $V$.

Let $u=\max\{i\in[t]:m_i>1\}$.
Since $\F_q^{n_i\times m_i}=\F_q$ for any $i>u$, then
    \begin{equation*}
        V=\F_q^{n_1\times m_1}\times\ldots\times\F_q^{n_u\times m_u}\times\underbrace{\F_q\times\ldots\F_q}_{t-u}=\F_q^{n_1\times m_1}\times\ldots\times\F_q^{n_u\times m_u}\times\F_q^{t-u}
    \end{equation*}
and for $X=(X_1,\ldots,X_t),Y=(Y_1,\ldots,Y_t)\in V$ we have
    \begin{equation*}
        \srkd(X,Y)=\sum_{i=1}^u\rkd(X_i,Y_i)+\Hd((X_{u+1},\ldots,X_t),(Y_{u+1},\ldots,Y_t))
    \end{equation*}
Moreover $$\langle X,Y\rangle=\sum_{i=1}^u\Tr(X_iY_i^t)+\langle(X_{u+1},\ldots,X_t),(Y_{u+1},\ldots,Y_t)\rangle,$$
where $\langle(X_{u+1},\ldots,X_t),(Y_{u+1},\ldots,Y_t)\rangle$ denotes the standard inner product of $(X_{u+1},\ldots,X_t)$ and $(Y_{u+1},\ldots,Y_t)$ in $\F_q^{t-u}$.
\end{example}

\begin{definition}
Let $(V,\dd)$ be a metric vector space and let $x\in V$,
The \textit{weight} of $x$ is defined as $\w(x)=\dd(x,0)$, where $0$ is the zero of $V$.
\end{definition}

In specific instantiations, we will let the notation for the weight follow the notation for the distance. For example, the weight associated to the Hamming metric $\Hd$ will be denoted by $\Hw$.

The notion of isometry allows us to identify codes that are equivalent for coding theory purposes.

\begin{definition}\label{def:isometry}
Let $\code,\code'$ be codes in a metric vector space  $(V,\dd)$.
We say that $\code$ is \textit{isometric} or {\em equivalent} to $\code'$ (notation $\code\simeq\code'$) if there is a linear map $\phi:V\rightarrow V$ such that $\phi(\code)=\code^\prime$ and $\dd(x,y)=\dd(\phi(x),\phi(y))$ for all $x,y\in V$.
Isometries induce an equivalence relation on the set of all codes, also known as \textit{code equivalence}.
An \textit{invariant} of a code $\code$ is any parameter which is constant on the elements of every equivalence class of $\code$.
\end{definition}

For example, the dimension of a code is an invariant, since isometries are bijective.
Another key invariant in coding theory is the minimum distance of a code.

\begin{definition}
Let $\code$ be a linear code in $(V,\dd)$, $\code\neq 0$. Its \textit{minimum distance} is
\begin{equation*}
d(\code)=\min\{\dd(c,c')\mid c,c'\in\code,c\neq c'\}.
\end{equation*}
\end{definition}

The minimum distance is directly connected to the error-correcting performance of the code, and it is an invariant because isometries preserve distances.

In this work, we focus on generalized weights, a family of code parameters. We use the term generalized weight to refer to any of these quantities and specialize it to any specific construction by adding other adjectives. For example, generalized Hamming (GH) weights are generalized weights in the Hamming metric space, whose study was initiated in~\cite{wei1991generalized}.
The original definition of GH weights makes use of the Hamming support of a code in the Hamming space.

\begin{definition}\label{def:Hamming_support}
    The \textit{Hamming support} of a vector $x=(x_1,\ldots,x_n)\in\F_q^n$ is the subset ${\supp(x)=\{i\in[n]:x_i\neq0\}}\subseteq[n]$.
    The support of a code $\code\leq\F_q^n$ is the union of the supports of its codewords, namely
    \begin{equation*}
        \supp(\code)=\bigcup_{c\in\code}\supp(c).
    \end{equation*}
\end{definition}

Notice that the support of a $1$-dimensional code coincides with the support of any of its generators, which justifies the use of the notation $\supp(\cdot)$ for both vectors and codes.
The following is the original definition of GH weights given in~\cite{wei1991generalized}.

\begin{definition}\label{def:GHW}
    Let $0\neq \code\leq\F_q^n$ be a $k$-dimensional linear code.
    For every $r=1,\ldots,k$ the $r^{th}$ \textit{generalized Hamming weight} is defined as
    \begin{equation*}
        d_r(\code)=\min\{|\supp(\Dcode)|:\Dcode\leq\code,\dim(\Dcode)=r\}.
    \end{equation*}
\end{definition}

GH weights are invariants of linear codes that capture many nontrivial properties, including the information leakage of coset coding schemes on the type II wire tap channel model.
In the case of the rank metric, two definitions of generalized weights with different characteristics have been proposed.
The definition of~\cite[Definition 10]{martinez2017relative} does not always yield a set of invariants, whereas~\cite[Definition 23]{ravagnani2016generalized} does. 
However, the first definition measures universal security on wire-tap networks.
Both families of generalized weights satisfy a duality statement.
In this paper, we introduce a general construction of generalized weights that will help us to understand the common characteristics of these weights, especially their duality theory.

Given a set of parameters defined for all codes in a space $V$, we say that they satisfy a duality statement if the parameters of $\code^\perp$ are determined by those of $\code$.
For example, the following is the well-known duality statement for GH weights.

\begin{theorem}[{\cite[Theorem 3]{wei1991generalized}}]\label{thm:duality_GHW}
Let $\code\leq\F_q^n$ be a $k$-dimensional code. Then
\begin{equation*}
\{d_s(\code^\perp)\mid1\leq s\leq n-k\}=[n]\setminus\{n+1-d_r(\code) \mid 1\leq r\leq k\}.
\end{equation*}
\end{theorem}

\subsection{Bounds on code parameters}
    
Determining the code parameters for a given code can be a hard task. In this context, bounds are useful, as they restrict the range of possible values of a parameter given the value of the others.
The Singleton Bound is one of the fundamental bounds in the Hamming metric. It relates the length, dimension, and minimum distance of a code.

\begin{theorem}[Singleton Bound, {\cite[Theorem 1]{singleton1964maximum}}]\label{thm:Singleton_Hamming}
Let $\code\leq(\F_q^n,\Hd)$ be a $k$-dimensional code with minimum distance $d(\code)=d$. Then $d+k\leq n+1$. 
\end{theorem}

\begin{definition}
A code $\code\leq\F_q^n$ is \textit{maximum distance separable} (MDS) if either $\code=0$ or its parameters meet the Singleton Bound.
For any $n,q$, we denote the family of MDS codes in $(\F_q^n,\Hd)$ by $\MDSFam$.
\end{definition}

Having the largest possible minimum distance among codes of a given dimension, MDS codes are considered optimal from the point of view of their error-correction capacity.
Moreover, for every $n$ and $q$, the family $\MDSFam$ is closed under duality by~\cite[Theorem 2.4.3]{huffman2010fundamentals}.
Some codes are MDS in a very obvious way: For every vector $v\in\F_q^n$ such that $\Hw(v)\in\{0,n\}$, the code generated by $v$ is MDS.
These are all MDS codes of dimension 0 and 1, and their duals are all MDS codes of dimension $n-1$ and $n$.
Such codes are referred to as trivial MDS codes.

A limitation to the use of MDS codes comes from the relation between the field size $q$ and their length $n$, which is subject to the following conjecture.

\begin{conjecture}[{\cite[Problem $I_{r,q}$]{segre1955p}}]\label{conj:MDS}
When $q$ is odd, the family \smash{$\MDSFam$} contains nontrivial codes if and only if $n\leq q+1$.
When $q$ is even, the family $\MDSFam$ contains nontrivial codes if and only if $n\leq q+2$. Furthermore, the only nontrivial codes for $n=q+2$ have dimension $k=3,q-1$.
\end{conjecture}

A first result which motivates the conjecture may be found in~\cite{segre1955}.
A proof of the conjecture in several cases, including when $q$ is a prime, may be found in~\cite{ball2012sets}.

The next result is the analog of the Singleton Bound in the context of rank-metric codes.

\begin{theorem}[{\cite[Theorem 5.4]{delsarte1978bilinear}}]\label{thm:Singleton_rank}
Let $n\leq m$ and let $\code\leq(\F_q^{n\times m},\rkd)$, $\code\neq 0$, be a $k$-dimensional code with minimum distance $d(\code)=d$. Then $k\leq m(n-d+1)$.
\end{theorem}

\begin{definition}
A code $\code\leq\F_q^{n\times m}$ whose parameters meet the previous bound is called \textit{maximum rank distance (MRD) code}.
For any $n,m,q$ the family of MRD codes in $\F_q^{n\times m}$ will be denoted by $\MRDFam$.
\end{definition}

Similarly to MDS codes, the family $\MRDFam$ is closed under duality by~\cite[Theorem 5.5]{delsarte1978bilinear}.
A notable difference between MDS and MRD codes is that, for any $n,m,q$, MRD codes of dimension $k$ exist in $\F_q^{n\times m}$ for every possible value of $k\leq nm$, as shown in~\cite{delsarte1978bilinear}. 

Another interesting family of bounds are the anticode bounds, which relate the dimension and maximum weight of a code.
The following are the statements of the Anticode Bound in the Hamming, rank, and sum-rank metric.

\begin{theorem}[{\cite[Proposition 6]{ravagnani2016generalized}}]\label{thm:anticode_Hamming}
    Let $\code\leq(\F_q^n,\Hd)$ be a code of $\dim(\code)=k$. Then ${k\leq\max\{\Hw(c)\mid c\in\code\}}$.
\end{theorem}

\begin{theorem}[{\cite[Proposition 47]{ravagnani2016rank}}]\label{thm:anticode_rank}
    Let $1\leq n\leq m$ and let $\code\leq(\F_q^{n\times m},\rkd)$ be a code of $\dim(\code)=k$. Then ${k\leq m\max\{\rkw(c)\mid c\in\code\}}$.
\end{theorem}

\begin{theorem}[\text{{\cite[Theorem 3.1]{camps2022optimal}}}]\label{thm:anticode_sum_rank}
    Let $\code\leq(\F_q^{n_1\times m_1}\times\ldots\times\F_q^{n_t\times m_t}),\srkd)$ be a code of $\dim(\code)=k$. Then ${k\leq\max\{\sum_{i=1}^tm_i\rank(c_i)\mid c=(c_1,\ldots,c_t)\in\code\}}$. In particular, if $m_1=\ldots=m_t=m$, then $k\leq m\max\{\srkw(c)\mid c\in\code\}$.
\end{theorem}

Optimal anticodes, i.e., codes that meet the Anticode Bound, play an important role in the context of generalized weights~\cite{ravagnani2016generalized,camps2022optimal}.

\begin{definition}\label{def:optimal_anticode}
    A code $\code\leq\F_q^n$ that attains the Anticode Bound of~\Cref{thm:anticode_Hamming} with equality is called an \textit{optimal anticode (OA)} in the Hamming metric.
     Analogous notions are defined in the rank and sum-rank metric, using the respective anticode bounds.
\end{definition}

Given a vector space $V$ with coordinates, the standard basis of $V$ is the set of vectors whose components are equal to 0, except for one equal to 1.

\begin{definition}\label{def:SOA}
Let $(V,\dd)$ be a metric vector space where we have an anticode bound, then a \textit{standard optimal anticode} (SOA) is an OA that is isometric to a code generated by standard basis vectors.
\end{definition}

\begin{example}\label{ex:standard_HSA}
For $q\neq 2$, every OA in the Hamming space $(\F_q^n,\dd_H)$ is a SOA. A proof of this fact may be found in~\cite[Proposition 9]{ravagnani2016generalized}, where SOAs are called \textit{free codes}.
The smallest example of a binary OA that is not an SOA is $\Acode=\langle (1,0,1), (0,1,1)\rangle$, the single parity check code of length~3.
SOAs in the Hamming space can also be described as support spaces:
For $i\in[n]$, let $e_i$ be the $i^{th}$ standard basis vector and let $I\subseteq[n]$. Then
\begin{equation}\label{eq:free_anticode}
\Acode_I=\{x\in\F_q^n:\supp(x)\subseteq I\}=\langle e_i:i\in I\rangle
\end{equation}
is an SOA, and all SOAs are of this form.
In~\Cref{subsec:binary_OA} we show that for $q=2$ any optimal anticode is either a standard optimal anticode, or a single parity check code of odd length, or a direct sum of one or more codes of each of these two kinds.
\end{example}

\begin{example}\label{ex:standard_OA_rank}
Let $0\leq t\leq n\leq m$ and consider
\begin{equation}\label{eq:standard_OA_rank}
    \Acode_t=\{X\in\F_q^{n\times m}:\textnormal{the last $n-t$ rows of $X$ are zero}\}=\langle E_{ij}\mid 1\leq i\leq t, 1\leq j\leq m\rangle,
\end{equation}
where $E_{ij}$ denotes the matrix whose entries are all 0, except for a 1 in position $(i,j)$. It can be shown that every OA of maximum rank $t$ in $\F_q^{n\times m}$ is isometric to $\Acode_t$. In particular, all OAs in the rank metric are SOAs, see e.g.~\cite[Theorem 11.3.15]{gorla2021rank}.
\end{example}

\begin{example}\label{ex:standard_OA_sumrank}
Let $V=\F_q^{n_1\times m_1}\times\ldots\times\F_q^{n_u\times m_u}\times\F_q^{t-u}$, where $1\leq n_i\leq m_i$ for all $i=1,\ldots,u$, and $m_1\geq m_2\geq\ldots\geq m_u\geq 2$. It follows from~\cite[Theorem III.11]{camps2022optimal} that every OA $\Acode$ in $V$ has the form $\Acode=\Acode_1\times\ldots\Acode_u\times\Acode^\prime$, where $\Acode_i$ is an OA in $\F_q^{n_i\times m_i}$ for $1\leq i\leq u$ and $\Acode^\prime$ is an OA in $\F_q^{t-u}$. By the definition of SOA, $\Acode$ is an SOA if and only if each of its direct summands is. It follows from the discussion in Example~\ref{ex:standard_HSA} and Example~\ref{ex:standard_OA_rank} that, for every $q\neq 2$, every OA in the sum-rank metric is an SOA. For $q=2$ and any $t,u$ such that $t-u\geq 3$, one can easily produce an example of an OA which is not an SOA by taking the product of SOAs in the rank-metric with an OA in $\F_2^{t-u}$ which is not an SOA.
\end{example}

\section{Generalized weights and test families}\label{sec:general_theory}

In this section, we define generalized weights in an arbitrary vector space $V$ using intersections with a chosen family of codes.
In general, a family of codes is just a set of subcodes of $V$.
A test family is a family of codes that satisfies some extra conditions, which we show to be sufficient to grant a duality statement for the weights.
The following definition lists these conditions. As will be clear by looking at the examples, they are satisfied by many families that have been considered in the literature.
A simple yet effective approach to the construction of test families using chains of codes is outlined.
As a running example, we instantiate our constructions in the Hamming metric space to recover generalized Hamming weights using SOAs as a test family.
In~\cite{ravagnani2016generalized} it  was proven that SOAs coincide with OAs in the Hamming metric for $q\neq2$, while for $q=2$ the class of OAs is strictly larger. No classification of binary OAs was provided. We classify binary OAs at the end of this section.

\begin{definition}\label{def:test_family} 
Let $\ell\geq 1$ be an integer that divides $\nu=\dim(V)$.
A nonempty family of codes $T$ in $V$ is a \textit{test family} with \textit{step} $\ell$ if it satisfies the following conditions:
\begin{enumerate}
\item $\ell|\dim(\Tcode)$ for every $\Tcode\in T$,
\item if $\Tcode\in T$, then $\Tcode^\perp\in T$,
\item every $\Tcode\in T$ of dimension $a\ell>0$ contains a code $\Tcode'\in T$ such that $\dim(\Tcode')=(a-1)\ell$,
\item every $\Tcode\in T$ of dimension $a\ell<\nu$ is contained in a code $\Tcode'\in T$ such that $\dim(\Tcode')=(a+1)\ell$.
\end{enumerate}
\end{definition}

The elements of $T$ are called \textit{test codes}. 
It can be checked that the zero code and the whole space $V$ are test codes for every $T$.
Examples of test families are easy to find for any vector space. We start with a trivial example.

\begin{example}
For any vector space $V$ and $\ell\mid\nu=\dim(V)$, the set of all subspaces of $V$ whose dimension is divisible by $\ell$ is a test family.
\end{example}

\begin{example}\label{ex:anticodes_H1}
Let $V=\F_q^n$ be the Hamming-metric space.
The family of OAs and that of SOAs in $V$ are test families with step $\ell=1$.
\end{example}

Given a code $\code$ and a test family $T$, we define the generalized weights of $\code$ with respect to the family $T$, following the rationale of~\cite{ravagnani2016generalized}.

\begin{definition}\label{def:T-weights}
Let $\code\leq V$ be a $k$-dimensional code, and let $T$ be a test family with step $\ell$.
For $r=1,\ldots,k$, the $r^{th}$ \textit{generalized $T$-weight} of $\code$ is the quantity
\begin{equation*}
\tau_{T,r}(\code)=\min\{\dim(\Tcode)\mid\Tcode\in T,\dim(\Tcode\cap\code)\geq r\}.
\end{equation*}
We omit $T$ from the notation whenever it is clear from context and call these parameters simply generalized weights.
The sequence of all weights $$\tau(\code)=(\tau_1(\code),\ldots,\tau_k(\code))=(\tau_r(\code)\mid r=1,\ldots,k)$$ is called \textit{$T$-weight hierarchy}. For a fixed $0\leq h<\ell$, the subsequence of weights whose index $r$ is congruent to $h$ modulo $\ell$  is denoted by $\tau^h(\code)$, i.e.,
   $$\tau^h(\code)=(\tau_r(\code)\mid r=1,\ldots,k\textnormal{ and }r=h\mod\ell).$$
\end{definition}

Notice that if $\ell$ is the step of the test family, then $[\nu]_\ell$ is the set of possible values for the weights.

\begin{remark}
\Cref{def:T-weights} could be extended to include families $T$ of codes that do not meet the requirements of test families.
We explore this idea in~\Cref{sec:MDS}, where we study the weights obtained using the families of MDS and MRD codes to define weights in the Hamming and rank metrics.
\end{remark}

\begin{definition}
Let $\code\leq V$ be a code.
We say that $\Tcode\in T$ \textit{realizes} the weight $\tau_r(\code)$ if $\dim(\Tcode\cap\code)\geq r$ and $\dim(\Tcode)=\tau_r(\code)$. 
\end{definition}

The definition of test family and that of generalized weights do not depend on the choice of a distance function. In that sense, our approach is distance-free. Moreover, the $T$-weights are not necessarily code invariants. However, this is the case if the test family is closed under isometries.

\begin{definition}
Let $(V,\dd)$ be a metric vector space.
A test family $T$ is called \textit{metric} if it is closed under isometry. In other words, $T$ is metric if for all test codes $\Tcode\in T$ and all $\Tcode'\simeq\Tcode$ we have $\Tcode'\in T$.
\end{definition}

OAs and SOAs are examples of metric test families in the Hamming, rank, and sum-rank metrics. The classical definition of generalized weights for these codes often coincides with that of generalized $T$-weights for $T$ the family of SOAs. Since the dimension of an SOA coincides with the weight of its support (possibly up to a constant), in this case generalized $T$-weights are related to the distance. In addition, they turn out to be code invariants.
In fact, metric test families always give rise to invariants.

\begin{lemma}\label{lem:metric_T}
The generalized $T$-weights are invariants of the subcodes of $V$ if and only if the test family $T$ is metric.
\end{lemma}

\begin{proof}
Let $T$ be metric and let $\code,\code'\leq V$ be codes such that $\code'\simeq\code$ via an isometry $\phi$. Let $1\leq r\leq k=\dim(\code)$ and let $\Tcode\in T$ be a code realizing $\tau_r(\code)$.
Let $\Tcode'=\phi(\Tcode)$, then since $T$ is metric we have $\Tcode'\in T$. Moreover $\dim(\Tcode'\cap\code')=\dim(\phi(\Tcode\cap\code))\geq r$, implying $\tau_r(\code')\leq \tau_r(\code)$.
The same reasoning can be applied in the opposite direction to yield $\tau_r(\code)\leq \tau_r(\code')$, i.e., the generalized $T$-weights of $\code$ and $\code'$ coincide.
For the converse implication, assume that $T$ is not metric. Then there are two isometric codes of some dimension $k$, of which only one is in $T$.
By definition, the $k^{th}$ $T$-weights of these two codes cannot coincide.
\end{proof}

Test families can be merged to form larger families, in a sense made precise by the following lemma. The proof is straightforward.

\begin{lemma}\label{lem:union_families}
The union of test families is a test family. The union of metric test families is a metric test family.
\end{lemma}

The proof of the next lemma is also straightforward, under the assumption that the bilinear form on $V_1\times V_2$ is induced by those on $V_1$ and $V_2$. The lemma will be useful in the study of generalized weights in the sum-rank metric.

\begin{lemma}\label{lem:product_families}
Let $V_1$ and $V_2$ be vector spaces, and let $T_1$ and $T_2$ be test families with step $\ell$ in $V_1$ and $V_2$, respectively. Then $$T_1\times T_2=\{\mathcal{T}_1\times\mathcal{T}_2\mid \mathcal{T}_1\in T_1, \mathcal{T}_2\in T_2\}$$ is a test family with step $\ell$ in $V_1\times V_2$.
\end{lemma}

Finally, some subfamilies of test families are also test families.

\begin{lemma}\label{lem:multiple_families}
Let $(V,d)$ be a metric vector space and let $T$ be a test family with step $\ell$ in~$V$. Let $m$ be such that $\ell\mid m$ and $m\mid\nu=\dim(V)$. Then the family $\{\mathcal{T}\in T : m\mid\dim(\mathcal{T})\}$ is a test family with step $m$ in $V$. Moreover, it is metric if $T$ is metric.
\end{lemma}

A simple way to build a metric test family is to make use of chains of codes.

\begin{definition}
A \textit{chain} of codes in $V$ is a sequence of codes $C=(C_0,\ldots,C_m)$ such that $C_i\subseteq C_{i+1}$ for all $i=0,\ldots,m-1$.
A \textit{fixed step chain} is a chain $C$ such that $\nu=m\ell$ and $\dim(C_i)=i\ell$ for all $i=0,\ldots,m$.
The quantity $\ell$ is called the {\em step} of the chain.
The \textit{dual chain} of $C$ is the chain $C^\perp=(C_m^\perp,\ldots,C_0^\perp)$.
\end{definition}

Notice how the inclusions are reversed in the dual chain and that the dual chain of a fixed step chain is again a fixed step chain.

\begin{definition}\label{def:test_family_chain}
Let $m,\ell\geq 1$ be positive integers, $\nu=m\ell$, and consider a fixed step chain of codes $C=(C_0,\ldots,C_m)$ with step $\ell$.
The test family associated to $C$ is the family $T(C)$ of codes that are equivalent to a code in $C$ or $C^\perp$, that is,
\begin{equation*}
T(C)=\{\Tcode\subseteq V: \Tcode\simeq C_i\textnormal{ or }\Tcode\simeq C_i^\perp \textnormal{ for some $i$}\}.
\end{equation*}
\end{definition}

It is easy to check that, given a fixed step chain $C$ with step $\ell$, the family $T(C)$ is a metric test family with step $\ell$.
We revisit~\Cref{ex:anticodes_H1} from this point of view.

\begin{example}\label{exa:anticodes_H2}
Let $V=\F_q^n$, for $i=0,\ldots,n$ let $\Acode_i=\{x\in\F_q^n:\supp(x)\subseteq[i]\}$. The chain $A=(\Acode_0,\ldots,\Acode_n)$ gives rise to the test family $T(A)$, which is the family of SOAs described in~\Cref{ex:anticodes_H1}.
Notice that in this case the dual chain is redundant, as $\Acode_i^\perp\simeq\Acode_{n-i}$.
\end{example}

Chains are an elegant means of constructing metric test families.
However, there are quantities related to codes that, despite not being invariants, are relevant to applications and have interesting algebraic properties.
A relevant example are the generalized matrix weights defined in~\cite{martinez2017relative}.
Defining and studying these quantities using (non-metric) test families of codes allows us to give a unitary treatment that highlights their common algebraic features.

\subsection{Properties and duality}

Throughout this section, we fix a test family $T$ with step $\ell$ in $V$ and study the algebraic properties of generalized $T$-weights.
The main result is a general duality statement that relates the $T$-weights of $\code$ to the $T$-weights of $\code^\perp$.
This result unifies the duality theory of GH weights~\cite{wei1991generalized}, GD weights~\cite{ravagnani2016generalized}, and GM weights~\cite{martinez2017relative}. Moreover, it provides key insights that allow us to extend the duality of sum-rank metric codes established in~\cite{camps2022optimal}.
We start by proving some basic properties.

\begin{lemma}\label{lem:T-weights_basic_prop}
Let $\code$ be a $k$-dimensional code in $V$. The $T$-weights form a weakly increasing sequence and the weights in each sequence $\tau^h(\code)$ form a strictly increasing sequence for every $h=0,\ldots,\ell-1$.
Furthermore, 
\begin{equation*}
\bigg\lceil \frac{r}{\ell}\bigg\rceil\ell\leq\tau_r(\code)\leq\bigg\lceil\frac{\nu-k+r}{\ell}\bigg\rceil\ell
\end{equation*} 
for every $1\leq r\leq k$.
\end{lemma}

\begin{proof}
The weights in $\tau(\code)$ form a weakly increasing sequence because the sets of codes over which the intersections are computed form a weakly ascending chain.
Let $0\leq h<\ell$ and consider the sequence $\tau^h(\code)$. Let $\Tcode\in T(S)$ be a code that realizes $\tau_r(\code)$ for some $r=h\mod\ell$.
By~\Cref{def:T-weights}, the code $\Tcode$ contains a subcode $\Tcode'$ of codimension $\ell$, for which 
$\dim(\Tcode'\cap\code)\geq\dim(\Tcode\cap\code)-\ell\geq r-\ell$.
It follows that $$\tau_{r-\ell}(\code)\leq\dim(\Tcode')=\dim(\Tcode)-\ell<\dim(\Tcode)=\tau_r(\code).$$ This proves that $\tau^h(\code)$ is a strictly increasing sequence.

Finally, let $1\leq r\leq k$. The lower bound on $\tau_r(\code)$ follows directly from the definition.
To show the upper bound, notice that if $\Tcode\in T(S)$ has dimension $\dim(\Tcode)\geq\nu-k+r$, then we have $\dim(\code\cap\Tcode)\geq r$.
\end{proof}

This result plays the same role as~\cite[Theorem 1]{wei1991generalized} and~\cite[Theorem 30]{ravagnani2016generalized} in the respective context.

\begin{remark}
Special cases of \Cref{lem:T-weights_basic_prop} have appeared in several different contexts. If one chooses as test family the family of SOAs, \Cref{lem:T-weights_basic_prop} recovers the known inequalities on GH weights (implicit in~\cite{wei1991generalized}), GD weights~(\cite[Theorem 30]{ravagnani2016generalized}), and generalized weights of vector rank-metric codes~(\cite[Corollary 15]{kurihara2015relative}). If one chooses as test family the family of rank-support spaces, it recovers the inequalities for GM weights~(\cite[Section VII.A]{martinez2017relative}).
Finally, if one chooses as test family the family of SOAs whose dimension is a multiple of $m$, it recovers a special case of the inequalities for sum-rank metric codes~(\cite[Lemma V.5 and Proposition V.6]{camps2022optimal}). 
\end{remark}

In order to show that the $T$-weights satisfy a duality statement, we will use the following lemma.

\begin{lemma}\label{lem:T_duality_ineq}
Let $\code\subseteq V$ be a $k$-dimensional code, and let $k^\perp=\dim(\code^\perp)=\nu-k$.
Then 
\begin{align*}
\tau_{k^\perp+r-\tau_{r}(\code)}(\code^\perp)\leq \nu-\tau_{r}(\code)\quad&\text{for every $r\in[k]$,}\\
\tau_{k+s-\tau_{s}(\code^\perp)}(\code)\leq \nu-\tau_{s}(\code^\perp)\quad&\text{for every $s\in[k^\perp]$}.
    \end{align*}
\end{lemma}

\begin{proof}
We prove only the first inequality, as the second follows from applying the first inequality to the dual code.
Let $\Tcode$ be a code that realizes $\tau_{r}(\code)$. Then 
\begin{align*}
\dim(\code^{\perp}\cap\Tcode^\perp)&=k^\perp-\dim(\Tcode)+\dim(\Tcode\cap\code)\geq k^\perp-\tau_{r}(\code)+r.
\end{align*}
It follows that $\Tcode^\perp$ is one of the codes to be considered in determining $\tau_{k^\perp-\tau_{r}(\code)+r}(\code^\perp)$.
Hence
\begin{equation*}
\tau_{k^\perp-(\tau_{r}(\code)-r)}(\code^\perp)\leq\dim(\Tcode^\perp)=\nu-\tau_{r}(\code).
\end{equation*}
\end{proof}

The following is the main result of this paper. It describes how the $T$-weights of a code determine the $T$-weights of its dual.

\begin{theorem}\label{thm:duality}
Let $T$ be a test family with step $\ell$ in $V$. 
For any $0\leq h\leq\ell-1$, the sequences $\tau^h(\code^\perp)$ and $\tau^{h+k}(\code)$ determine each other. More precisely, we have 
\begin{equation*}
\{\tau_s(\code^\perp)\;:\;s=h\mod\ell\}=[\nu]_\ell\setminus\{\nu+\ell-\tau_{r}(\code)\;:\; r=k+h\mod\ell\}.
\end{equation*}
\end{theorem}

\begin{proof}
Fix $0\leq h\leq \ell-1$ and let $r\in[k]$, $s\in[k^\perp]$ with $r=h+k\mod\ell$ and $s=h\mod\ell$.
We claim that
\begin{equation}
\tau_{s}(\code^\perp)\neq \nu+\ell-\tau_{r}(\code).
\end{equation}
In fact, suppose by contradiction that the equality holds for some $r,s$ as above.
By the previous lemma we have $\tau_s(\code^\perp)\leq \nu-\tau_{k+s-\tau_s(\code^\perp)}(\code)$. Therefore, we must have $r>k+s-\tau_s(\code^\perp)$.
Since $r=k+s\mod\ell$ and $\ell\mid \tau_s(\code^\perp)$, then \begin{equation}\label{eqn:r}
r\geq k+s+\ell-\tau_s(\code^\perp).
\end{equation}
On the other hand, by the previous lemma 
$\tau_{s}(\code^\perp)>\tau_{s}(\code^\perp)-\ell= \nu-\tau_r(\code)\geq\tau_{k^\perp+r-\tau_{r}(\code)}(\code^\perp)$,
hence 
\begin{equation}\label{eqn:s}
s>k^\perp+r-\tau_{r}(\code).
\end{equation}
Combining (\ref{eqn:r}) and (\ref{eqn:s}), one gets
$$r>k+k^\perp+\ell+r-\tau_{r}(\code)-\tau_s(\code^\perp)=r,$$
which is a contradiction.
It follows that the sets
\begin{equation*}
\{\tau_s(\code^\perp)\;:\;s=h\mod\ell\}\text{ and }\{\nu-\tau_r(\code)+\ell\;:\;r=h+k\mod\ell\}
\end{equation*}
have an empty intersection. Since by \Cref{lem:T-weights_basic_prop} their cardinalities sum to $\nu/\ell$, they are one the complement of the other in $[\nu]_{\ell}$.
\end{proof}

\begin{example}\label{ex:duality_GHW}
Defining generalized $T$-weights with respect to the family of SOAs yields GH weights, as was proved in~\cite[Theorem 10]{ravagnani2016generalized}.
The restriction $q\geq 3$ can be lifted due to the fact that we use SOAs instead of OAs in our definition.
Since $\ell=1$, there is only one sequence $\tau^h(\code)$, i.e., $\tau(\code)=\tau^0(\code)$ for every $\code\leq\F_q^n$.
Hence we get
\begin{equation*}
\{\tau_s(\code^\perp)\mid1\leq s\leq n-k\}=\{n+1-\tau_r(\code)\mid1\leq r\leq k\}
\end{equation*}
which recovers~\Cref{thm:duality_GHW}.
\end{example}

\subsection{Classification of binary optimal anticodes}\label{subsec:binary_OA}

In this section, we classify binary OAs in the Hamming metric up to isometry. 

\begin{theorem}
Let $\Acode\leq\F_2^n$ be an OA of dimension $k$. Then either $\Acode$ is an SOA, or there exist non-negative integers $t,f$ and odd integers $n_1,\ldots,n_t$ such that $n=n_1+\ldots+n_t+f$ and
\begin{equation*}
\Acode\simeq\Acode_1\times\ldots\times\Acode_t\times\Acode_{\textnormal{H}}
\end{equation*}
where $\Acode_{\textnormal{H}}\leq\F_2^f$ is an SOA, and $\Acode_i\leq\F_2^{n_i}$ is a $[n_i,n_i-1]_2$ single parity check code for all $i=1,\ldots,t$.
\end{theorem}

\begin{proof}
Up to equivalence, we may assume that $\Acode$ has a systematic generator matrix ${G=(I_k \mid A)\in\F_2^{k\times n}}$, where $I_k\in\F_2^{k\times k}$ is the identity matrix and $A\in\F_2^{n-k\times k}$. We claim that every column of $A$ has even weight and every row of $A$ has weight zero or one. 
In fact, $(1,\ldots,1)G = (1, \ldots, 1 \mid (1, \ldots, 1) A)\in\Acode$ has weight $k+\Hw(1, \ldots, 1)A$. Since the maximum weight of $\Acode$ is $k$, $(1, \ldots, 1) A=(0\ldots0)$, i.e., every column of $\Acode$ has even weight.
Let $1\leq i\leq k$ and let $v_i\in\F_2^k$ be the vector whose entries are all equal to one, except for a zero in position $i$. Then $\Hw(v_iG)=k-1+\Hw(a_i)$, where $a_i$ denotes the $i^{th}$ row of $A$. Since the maximum weight of an element of $\Acode$ is $k$, the weight of the $i^{th}$ row of $\Acode$ is at most one.

Up to permuting the rows of $G$, we may suppose that the last $f$ rows of $A$ are zero and that each of the remaining rows has its only nonzero entry in the same column as the previous row or on the right of it. Up to permuting the columns of $G$, $G$ takes the form
    \begin{equation*}
        G=\begin{pmatrix}
            G_1 & 0 & \ldots & 0 & 0\\
            0 & G_2 & \ldots & 0 & 0\\
            \ldots & \ldots & \ldots & \ldots & \ldots\\
            0 & 0 & \ldots & G_t & 0\\
            0 & 0 & \ldots & 0 & G_{\textnormal{H}}
        \end{pmatrix}
    \end{equation*}
where $G_{\textnormal{H}}$ is an $f\times f$ identity matrix and $G_i=(I_{n_i-1}\mid \mathbf{1})\in\F_2^{(n_i-1)\times n_i}$, where $I_{n_i-1}$ is an $(n_i-1)\times (n_i-1)$ identity matrix and $\mathbf{1}$ is a column of ones. For $1\leq i\leq t$ let $\Acode_i\subseteq\F_2^{n_i}$ be the code with generator matrix $G_i$ and let $\Acode_{\textnormal{H}}$ be the code with generator matrix $G_{\textnormal{H}}$. Then $\Acode\simeq\Acode_1\times\ldots\times\Acode_t\times\Acode_{\textnormal{H}}$ and $n=n_1+\cdots+n_t+f$. Since both the dimension and the maximum weight of a code are additive on direct products, $\Acode_1,\ldots,\Acode_t,\Acode_{\textnormal{H}}$ are OAs. Therefore, $\Acode_{\textnormal{H}}$ is an SOA and $n_1,\ldots,n_t$ are odd. 
\end{proof}

\section{Generalized weights of rank-metric codes}\label{sec:rank}

In this section, the general framework developed in~\Cref{sec:general_theory} is instantiated with the appropriate test families to recover the two definitions of generalized weights in the rank metric: generalized Delsarte weights~\cite{ravagnani2016generalized} and generalized matrix weights~\cite{martinez2017relative}.
The different features of these two notions are nicely explained in terms of the test families used and~\Cref{thm:duality} gives a concise proof of the respective duality theorems.
All results on rank-metric codes stated without proof can be found in~\cite{gorla2021rank}.

In the end of the section, we discuss generalized weights for vector rank-metric codes, i.e., $\F_{q^m}$-linear rank-metric codes. Also in this context, we show how to recover the duality of generalized weights from our general result.

\subsection{Generalized Delsarte weights}\label{sec:GDW}

Let $n\leq m$ and let $V=\F_q^{n\times m}$ be endowed with the rank metric. 
The generalized Delsarte weights are invariants of rank-metric codes defined in~\cite[Definition 23]{ravagnani2016generalized} using optimal anticodes.
It is well-known that every optimal rank-metric anticode is a standard optimal anticode.

\begin{definition}[{\cite[Definition 23]{ravagnani2016generalized}}]\label{def:GDW}
Let $\code\leq \F_q^{n\times m}$ be a $k$-dimensional code. For $1\leq r\leq k$, the $r^{th}$ {\em generalized Delsarte (GD) weight} is
    \begin{equation*}
        a_r(\code)=\frac{1}{m}\min\{\dim(\Acode)\mid\Acode\textnormal{ is an optimal anticode}, \dim(\Acode\cap\code)\geq r\}.
    \end{equation*}
\end{definition}
To recognize GD weights as an example of $T$-weights, observe that the  family of OAs is the test family $T(A)$ obtained by using the fixed step chain $A=(\Acode_0=0,\Acode_1,\ldots,\Acode_n=V)$ as a starting point in~\Cref{def:test_family_chain}. Here, $\Acode_t$ is the SOA consisting of matrices whose last $n-t$ rows are zero.
Since $T(A)$ is a metric test family, \Cref{def:T-weights} yields a family of invariants of rank-metric codes. Moreover, these invariants coincide with the GD weights up to the constant factor $m$, since
\begin{equation}
    \tau_r(\code)=\min\{\dim(\Tcode)\mid\Tcode\in T(A),\dim(\Tcode\cap\code)\geq r\}=ma_r(\code)
\end{equation}
for all $r$. The results contained in~\cite[Section 6]{ravagnani2016generalized} therefore follow from~\Cref{thm:duality}. In particular, we recover the duality statement for GD weights.

\begin{theorem}
Let $T$ be the family of OAs in the rank-metric space $\F_q^{n\times m}$ and let $\code$ be a $k$-dimensional code.
The sequences $\tau^h(\code)$ and $\tau^{h+k}(\code^\perp)$ determine each other via
\begin{equation*}
\tau^{h}(\code)=[nm]_m\setminus\{m(n+1)-x:x\in\tau^{h+k}(\code^\perp)\}
\end{equation*}
for every $0\leq h<m$. In particular, $\tau(\code)$ determines $\tau(\code^\perp)$.
\end{theorem}

\subsection{Generalized matrix weights}

Consider again the rank metric space $\F_q^{n\times m}$.
Generalized matrix (GM) weights were defined in~\cite{martinez2017relative} to measure the leakage of a rank-metric code when used as a coset code in a wire-tap network. They are a family of generalized weights for rank-metric codes, which is different from GD weights. We refer to~\cite[Section~11.5]{gorla2021rank} for a discussion of the properties of both families of weights and the relation between the two. In particular, GM and GD weights coincide for $n<m$, see~\cite[VIII.C]{martinez2017relative}, while they do not coincide in general for $n=m$, see~\cite[Example~2.10]{gorla2020rankmetric}. For $n>m$, the GM weights of $\code$ coincide with the GD weights of the transposed code $\code^t\subseteq\F_q^{m\times n}$, see~\cite[Theorem XXX]{gorla2021rank}
Moreover, the GM weights are not invariants if $n=m$. 
It is shown in~\cite[Proposition 65]{martinez2017relative} that the GM weights satisfy a duality statement.
These differences and analogies between the GD weights and the GM weights become clear when interpreting them as generalized $T$-weights. The starting point for the definition of the GM weights is the family of rank-support spaces.

\begin{definition}
The {\em rank-support space} associated to $\Lcal\leq\F_q^n$ is
$$\Vcal_\Lcal=\{X\in \F_q^{n\times m}:\colsp(X)\leq\Lcal\},$$
where $\colsp(X)$ denotes the column space of the matrix $X$.  
\end{definition}

The definition of GM weights is as follows.

\begin{definition}[{\cite[Definition 10]{martinez2017relative}}]
Let $\code\leq\F_q^{n\times m}$ be a $k$-dimensional code. For $1\leq r\leq k$ the {\em $r^{th}$ generalized matrix (GM) weight} is
\begin{equation*}d_r(\code)=\min\{\dim(\Lcal)\mid\Lcal\leq\F_q^n,\dim(\Vcal_\Lcal\cap\code)\geq r\}.
\end{equation*}
\end{definition}

In our language, GM weights are generalized $T$-weights for $T$ the family of rank-support spaces. 
\begin{lemma}
The family of rank-support spaces $T=\{\Vcal_\Lcal:\Lcal\leq\F_q^n\}$ is a test family with step $m$. If $m=n$, then $T$ is not a metric family.
\end{lemma}

\begin{proof}
It is easy to check that $T$ is a test family for any $n,m$. If $m=n$, let $1\leq r< n$ and let 
$$\Acode_r=\{X\in\F_q^{n\times m}:\textnormal{the last $n-r$ rows of $X$ are zero}\}.$$  
Then $\Acode_r=V_{\langle e_1,\ldots,e_{r}\rangle}\in T$, $\Acode_r\simeq\Acode_r^t$, but $\Acode_r^t\notin T$, hence $T$ is not metric.
\end{proof}

Consider the generalized $T$-weights for $T$ the family of rank-support spaces.
Since $\dim(\Vcal_\Lcal)=m\dim(\Lcal)$, it follows that $\tau_r(\code)=m d_r(\code)$ for all $1\leq r\leq k$.
Therefore, \Cref{thm:duality} implies~\cite[Proposition 65]{martinez2017relative}. The duality statement for the generalized $T$-weights in this case reads as follows.

\begin{theorem}
Let $T$ be the family of rank-support spaces in $\F_q^{n\times m}$, and let $\code\leq\F_q^{n\times m}$ be a $k$-dimensional code.
For every $0\leq h<m$ the sequences $\tau^h(\code)$ and $\tau^{h+k}(\code^\perp)$ determine each other via
\begin{equation*}
\tau^{h}(\code)=[nm]_m\setminus\{m(n+1)-x:x\in\tau^{h+k}(\code^\perp)\}.
\end{equation*}
In particular, $\tau(\code)$ determines $\tau(\code^\perp)$.
\end{theorem}

\subsection{Generalized weights of $\F_{q^m}$-linear codes}

Let $n\leq m$ and let $V=\F_{q^m}^n$ be endowed with the rank metric. An $\F_{q^m}$-linear subspace of $V$ is called a vector rank-metric code.  We recall that the rank metric on $\F_{q^m}^n$ is defined as $$\rkd(u,v)=\dim(\langle u_1-v_1,\ldots,u_n-v_n\rangle_{\F_q})$$ for any $u=(u_1,\ldots,u_n), v=(v_1,\ldots,v_n)\in \F_{q^m}^n$. We refer to \cite{gorla2021rank} for the definition and basic properties of vector rank-metric codes in $\F_{q^m}^n$ and a discussion of their relation with $\F_{q^m}$-linear codes in $\F_q^{n\times m}$.

Many equivalent definitions were proposed for the generalized weights of vector rank-metric codes. We refer to \cite[Section 11.5]{gorla2021rank} for the list of definitions and a discussion on their equivalence.
Here, we present the definition in the form that is most convenient for our purposes.

\begin{definition}
Let $\code\leq \F_{q^m}^n$ be a $k$-dimensional code. For $1\leq r\leq k$, the $r^{th}$ {\em generalized rank weight} is
    \begin{equation*}
        w_r(\code)=\min\{\dim_{\F_{q^m}}(\Acode)\mid\Acode\textnormal{ is an optimal anticode}, \dim_{\F_{q^m}}(\Acode\cap\code)\geq r\}.
    \end{equation*}
\end{definition}

To recognize these weights as an example of $T$-weights, observe that the  family of OAs is a test family of step $\ell=1$ and coincides with the family of SOAs by \cite[Theorem 11.3.16]{gorla2021rank}.
Therefore, \Cref{def:T-weights} yields a family of invariants of vector rank-metric codes. These invariants coincide with generalized rank weights
\begin{equation}
    \tau_r(\code)=\min\{\dim(\Tcode)\mid\Tcode\in T(A),\dim(\Tcode\cap\code)\geq r\}=w_r(\code)
\end{equation}
for all $r$. In particular, \Cref{thm:duality} produces a duality statement for generalized rank weights of vector rank-metric codes.

\begin{theorem}
Let $T$ be the family of OAs in the rank-metric space $\F_{q^m}^n$ and let $\code$ be a $k$-dimensional code.
The sequences $\tau(\code)$ and $\tau(\code^\perp)$ determine each other via
\begin{equation*}
\tau(\code)=[n]\setminus\{n+1-x:x\in\tau(\code^\perp)\}.
\end{equation*}
\end{theorem}

\section{Generalized weights in the sum-rank metric}\label{sec:sumrank} 

In this section, we focus on the sum-rank space $$V=\F_q^{n_1\times m_1}\times\ldots\times\F_q^{n_t\times m_t}=\F_q^{n_1\times m_1}\times\ldots\times\F_q^{n_u\times m_u}\times\F_q^{t-u}$$ of~\Cref{ex:sumrank_space}. We assume without loss of generality that $1\leq n_i\leq m_i$ for all $i$ and that $m_1\geq\ldots\geq m_t$. Let $u=\max\{i\mid m_i\geq 2\}$. If $m_1=\ldots=m_t=1$, let $u=0$.

We know from~\cite[Theorem IV.11]{camps2022optimal} that an OA in $V$ is the product of an OA in $\F_q^{n_i\times m_i}$ for every $1\leq i\leq u$ and an OA in $\F_q^{t-u}$.
Let $S$ be the family of standard optimal anticodes in~$V$. Since every OA in the rank-metric is an SOA, we have
$$S=\{\Acode_1\times\ldots\times\Acode_t\mid\Acode_i\leq\F_q^{n_i\times m_i}\textnormal{ is an OA for every }i\in[t]\}.$$
If $t-u\geq 3$ and $q=2$, one can easily construct examples of OAs that are not contained in $S$, see also~\cite[Remark 5.2]{camps2022optimal}.
Generalized weights of sum-rank metric codes are defined in~\cite{camps2022optimal} using the family $S$. Notice that $S$ is not a test family in general, see Remark~\ref{rmk:srknottestfam}. 

\begin{definition}[\text{\cite[Definition VI.1]{camps2022optimal}}]\label{def:GSRW}
Let $\code\leq V$ be a $k$-dimensional code. For $r=1\ldots,k$ the {\em $r^{th}$ generalized sum-rank (GSR) weight} is
\begin{equation*}
d_r(\code)=\min\biggl\{\sum_{i=1}^t\frac{\dim(\Acode_i)}{m_i}:\Acode=\Acode_1\times\ldots\times\Acode_t\in S,\dim(\code\cap\Acode)\geq r\biggr\}.
\end{equation*}
\end{definition}

\begin{remark}\label{rmk:srknottestfam}
The family $S$ of SOAs in the sum-rank metric is not a test family in general. In fact, it contains codes of dimension $\sum_{i=1}^t r_im_i$ for every choice of $r_1,\ldots,r_t\geq 0$, therefore the only possible values for the step $\ell$ are the divisors of $\gcd\{m_1,\ldots,m_t\}$. However, if $\ell<m_1$, let $\Acode=0\times\F_q^{n_2\times m_2}\times\ldots\times\F_q^{n_t\times m_t}\in T$. Any $\Acode^\prime\in S$ that contains $\Acode$ has $\dim(\Acode^\prime)=\dim(\Acode)+hm_1$ for some positive integer $h$. Therefore, $S$ does not satisfy property 4. of a test family. We conclude that $S$ can be a test family only if $\ell=m_1=\ldots=m_t$.
\end{remark}

GSR weights are code invariants~\cite[Remark V.3]{camps2022optimal}, but they do not satisfy a duality statement analogous to the Wei duality of the Hamming or the rank metric, see~\cite[Example V.10]{camps2022optimal}.
However, duality holds in the case $m_1=m_2=\ldots=m_t=m$, as shown in~\cite[Theorem V.9]{camps2022optimal}.
Notice that, if $m=1$, then $V=\F_q^t$ is the Hamming space, and the GSR weights coincide with the GH weights. If $m>1$, then $V$ is a product of rank-metric spaces and $S$ contains all optimal anticodes. 

\paragraph{GSR weights as $T$-weights.}

In this section, we find a test family $T\subseteq S$, which allows us to recover and extend the duality statement from~\cite[Theorem V.9]{camps2022optimal}. 

\begin{definition}
A code $\Acode\leq V$ is an {\em $m$-standard optimal anticode ($m$-SOA)} if it is a standard optimal anticode and its dimension is divisible by $m$. We denote by $T$ the family of $m$-SOAs in $V$. 
\end{definition}

By~\cite[Theorem IV.11]{camps2022optimal}, every $m$-SOA has the form
\begin{equation} \label{eq:mSOA}  \Acode=\Acode_1\times\ldots\times\Acode_t\times\Acode_{H}
\end{equation}
where $\Acode_i\subseteq\F_q^{n_i\times m}$ is an optimal rank-metric anticode for all $i=1,\ldots,t$ and $\Acode_{H}\subseteq\F_q^\alpha$ is an SOA such that $m|\dim(\Acode_H)$. 

\begin{lemma}\label{lem:mSOA}
Assume that $m\mid t-u$. Then the family $T$ of $m$-SOAs in $V$ is a metric test family with step $m$.
\end{lemma}

\begin{proof}
The family of SOAs in $\F_q^{n_i\times m}$ coincides with that of OAs and is a metric test family, as discussed in Section~\ref{sec:GDW}. Hence, the family of $m$-sOAs in $\F_q^{n_i\times m}$ is a metric test family by Lemma~\ref{lem:multiple_families}. Combining (\ref{eq:mSOA}) and Lemma~\ref{lem:product_families}, one sees that $T$ is a test family. 
Using the structure of isometries in the sum-rank metric from~\cite[Theorem IV.2]{camps2022optimal}, one can easily check that $T$ is a metric family.
\end{proof}

Lemma~\ref{lem:mSOA} implies the validity of all the statements in~\Cref{sec:general_theory} for the corresponding $T$-weights. In particular, the duality statement reads as follows.

\begin{theorem}\label{thm:dualityGSRW}
Let $V=\F_q^{n_1\times m}\times\ldots\times\F_q^{n_u\times m}\times\F_q^{t-u}$, where $m\geq 2$, $0\leq u\leq t$, $1\leq n_i\leq m$ for all $i$, and $m\mid t-u$. Let $T$ be the family of $m$-SOAs in $V$. Let $\code\leq V$ be a code.
For any $0\leq h\leq m-1$, the sequences $\tau^h(\code^\perp)$ and $\tau^{h+k}(\code)$ determine each other. More precisely,
\begin{equation*}
\{\tau_s(\code^\perp)\;:\;s=h\mod m\}=[\nu]_m\setminus\{\nu+m-\tau_{r}(\code),r=k+h\mod m\}\;
\end{equation*}
where $\nu=\dim(V)=m(n_1+\ldots+n_u)+t-u$.
\end{theorem}

Notice that, if $u=t$, then $\code\leq\F_q^{n_1\times m}\times\ldots\times\F_q^{n_u\times m}$ and the $T$-weights are a rescaling of the GSR weights:
\begin{equation*}
\tau_r(\code)=\min\{\dim(\Tcode)\mid\Tcode\in T,\dim(\Tcode\cap\code)\geq r\}=md_r(\code)
\end{equation*}
for $1\leq r\leq\dim(\code)$.
It follows that~\Cref{thm:dualityGSRW} extends~\cite[Theorem 5.9]{camps2022optimal}.
 Moreover, notice that, if $u<t$ and $m\mid t-u$, then the $T$-weights of a code $\code\leq\F_q^{n_1\times m}\times\ldots\times\F_q^{n_u\times m}\times\F_q^{t-u}$ do not determine its GSR weights, and its GSR weights do not determine its $T$-weights. 

\begin{example}
Let $\code_1,\code_2,\code_3\leq V=\F_2^{2\times 2}\times\F_2^2$ be given by
$$\code_1=\Biggl\{\left(\begin{pmatrix}    
a&b\\0&0
\end{pmatrix},(c,c)\right)\mid a,b,c\in\F_2\Biggr\},\;\; 
\code_2=\Biggl\{\left(\begin{pmatrix}    
a&b\\0&0
\end{pmatrix},(c,0)\right)\mid a,b,c\in\F_2\Biggr\},$$ 
$$\code_3=\Biggl\{\left(\begin{pmatrix}    
a&b\\0&c
\end{pmatrix},(c,0)\right)\mid a,b,c\in\F_2\Biggr\}.$$
If $T$ is the family of $2$-SOAs in $V$, the generalized $T$-weights of $\code_1,\code_2,\code_3$ are $$\tau(\code_1)=\tau(\code_2)=(2,2,4)\; \mbox{ and }\;\tau(\code_3)=(2,2,6),$$ while their GSR weights are $$d(\code_1)=d(\code_3)=(1,1,3)\; \mbox{ and }\;d(\code_2)=(1,1,2).$$
\end{example}

We conclude the section with a toy example of the duality statement at work.

\begin{example}
The code $\code\leq V=\F_2^{2\times 2}\times\F_2^4$ given by
    \begin{equation*}
        \code=\Biggl\{\left(\begin{pmatrix}
            a&b\\0&0
        \end{pmatrix},(c,c,c,c)\right)\mid a,b,c\in\F_2\Biggr\},
    \end{equation*}
has dimension $k=3$ and $\nu=\dim(V)=8$. Let $T$ be the family of $2$-SOAs in $V$.
One can check that the generalized $T$-weights of $\code$ are $\tau(\code)=(2,2,6)$.
By~\Cref{thm:duality}
\begin{align*}
\{\tau_2(\code^\perp),\tau_4(\code^\perp)\}&=\{\tau_r(\code^\perp)\mid r=_20\}=[\nu]_2\setminus\{\nu+2-\tau_r(\code)\mid r=_2 3\}=\{2,6\}\\
\{\tau_1(\code^\perp),\tau_3(\code^\perp),\tau_5(\code^\perp)\}&=\{\tau_r(\code^\perp)\mid r=_21\}=[\nu]_2\setminus\{\nu+2-\tau_r(\code)\mid r=_2 4\}=\{2,4,6\}.
\end{align*}
Hence, the weight hierarchy of the dual code is $\tau(\code^\perp)=(2,2,4,6,6)$.
This may also be checked by directly computing the weights of $\code^\perp$.
\end{example}

\section{Generalized MDS and MRD weights}\label{sec:MDS}

In this section, we test the limits of the assumptions made to define the $T$-weights.
We study two classes of weights, one in the Hamming space and the other in the rank-metric space, defined using MDS and MRD codes instead of a test family in~\Cref{def:T-weights}.
We show that these weights are code invariants, but they do not satisfy a duality statement. Since MDS codes are a test family for $q\gg 0$, in that case the corresponding sequence of generalized weights satisfies a duality statement. In addition, when $q\gg 0$, these generalized weights are equivalent to the code distances, another code invariant defined in~\cite{camps2025code}.

\subsection{Generalized MDS weights}

Throughout this section, we let $V=\F_q^n$ be the Hamming-metric space, $n\leq q+1$. Let $M=\MDSFam$ be the family of MDS codes in $\F_q^n$. It is well-known that $M$ is nonempty and closed under duality and code equivalence. The assumption that $n\leq q+1$ guarantees that $M$ contains a code of dimension $k$ for all $0\leq k\leq n$. Notice that, for the values of $q$ for which~\Cref{conj:MDS} holds, for $n\geq q+2$ there are no MDS codes, or MDS codes of dimension $k$ exist only for a few values of $k$, making the next definition trivial or less interesting.

Similarly to~\Cref{def:T-weights}, we can associate a sequence of generalized weights to a code by intersecting it with MDS codes.

\begin{definition}\label{def:GMDSW}
Let $\code\leq\F_q^n$ be a $k$-dimensional code. For every $1\leq r\leq k$ the {\em $r^{th}$ generalized MDS (GMDS) weight} is 
\begin{equation*}
\mu_r(\code)=\min\{\dim(\Mcode)\mid\Mcode\in M,\dim(\Mcode\cap\code)\geq r\}.
\end{equation*}
Let $\mu(\code)$ denote the sequence of generalized MDS weights.
\end{definition}

The next lemma summarizes the main properties of these weights.

\begin{lemma}\label{lem:GMDSW_properties}
Let $\code\leq\F_q^n$ be a $k$-dimensional code, $k^\perp=n-k$. Then:
\begin{enumerate}
\item the GMDS weights are invariants of $\code$,
\item $r\leq\mu_r(\code)\leq n-k+r$ for all $r=1,\ldots,k$,
\item $\mu_1(\code)\leq\mu_2(\code)\leq\ldots\leq\mu_k(\code)$,
\item $\mu_{k^\perp+r-\mu_r(\code)}(\code^\perp)\leq n-\mu_r(\code)$ for all $r=1,\ldots,k$.
\end{enumerate}
\end{lemma}

\begin{proof}
1. The weights are invariants because the family $M$ is closed under isometry, so the argument of~\Cref{lem:metric_T} applies.\\
\noindent
2. The inequalities follow from the following two facts: If $\dim(\code\cap\Mcode)\geq r$, then $\dim(\Mcode)\geq r$. If $\dim(\Mcode)=n-k+r$, then $\dim(\code\cap\Mcode)\geq r$.\\
\noindent
3. Since $\dim(\code\cap\Mcode)\geq r$ implies $\dim(\code\cap\Mcode)\geq r-1$, then $\mu_{r-1}(\code)\leq\mu_r(\code)$.\\
\noindent
4. Let $\Mcode\in M$ be a code that realizes $\mu_r(\code)$, then the same reasoning as in the proof of~\Cref{lem:T_duality_ineq} shows that $\dim(\code^\perp\cap\Mcode^\perp)\geq k^\perp+r-\mu_r(\code)$. The inequality follows.
\end{proof}

The next example shows that, unlike generalized Hamming weights, GMDS weights do not form a strictly ascending sequence.

\begin{example}\label{ex:MDS_bad_nestability}
Let $q=7$, $n=6$, and consider the $\MDS$ code $\code\subseteq\F_7^6$ specified by the generator matrix
\begin{equation*}
    G=\begin{pmatrix}
        6 & 1 & 1 & 6 & 0 & 0\\
        1 & 1 & 6 & 6 & 2 & 3\\
        1 & 1 & 1 & 1 & 1 & 1\\
    \end{pmatrix}.
\end{equation*}
The projective points corresponding to the columns of $G$ form a complete arc in the 2-dimensional projective space PG(2,7) (see~\cite{hirschfeld1998projective}), implying that the dual code $\code^\perp$ cannot be embedded in an $\MDS$ supercode (see~\cite{zhang2019deep}).
It follows that $\code$ is a $[6,3,4]_7$ MDS code that does contain an $\MDS$ subcode of dimension 2.
Hence we have $\mu_2(\code)>2$ and $\mu_2(\code)\leq\mu_3(\code)=3$, which implies that $\mu_2(\code)=\mu_3(\code)=3$.
\end{example}

The example also shows that $M$ is not a test family.
The next proposition provides a restriction on the GMDS weights of a code, in terms of the GMDS weights of its dual. However, it is not a duality statement. The next example will show that the GMDS weights do not satisfy a duality statement. 

\begin{proposition}\label{lem:duality_GMDSW}
Let $\code$ be an $[n,k]_q$ code, then for all $t=1,\ldots,k$
\begin{equation*}
\mu_t(\code)\in [n]\setminus\{n+1-\mu_r(\code^\perp)\;:\;r=1,\ldots n-k\}.
\end{equation*}
\end{proposition}

\begin{proof}
Let $k^\perp=n-k$ and suppose by contradiction that there exist $t\in[k]$ and $r\in[k^\perp]$ such that $\mu_t(\code)+\mu_r(\code^\perp)=n+1$. 
By~\Cref{lem:GMDSW_properties}.4 
\begin{equation*}
\mu_t(\code)+\mu_{k^\perp+t-\mu_t(\code)}(\code^\perp)\leq n,
\end{equation*} hence $r\geq k^\perp+t+1-\mu_t(\code)$ by~\Cref{lem:GMDSW_properties}.3.
By~\Cref{lem:GMDSW_properties}.4 one also has
\begin{equation*} 
\mu_r(\code^\perp)+\mu_{k+r-\mu_r(\code^\perp)}(\code)\leq n,
\end{equation*} 
which by~\Cref{lem:GMDSW_properties}.3 implies that $t\geq k+r+1-\mu_r(\code^\perp)$.
Summing the two inequalities yields $t+r\geq n+t+r+2-\mu_t(\code)-\mu_r(\code^\perp)$, i.e., $\mu_t(\code)+\mu_r(\code^\perp)\geq n+2$, a contradiction.
\end{proof}

\begin{example}\label{ex:counterexample_duality_GMDSW}
Consider the $[8,4,3]_7$ code $\code_1$ and the $[8,4,4]_7$ code $\code_2$ generated respectively by
    \begin{equation*}
        G_1=\begin{pmatrix}
        1 & 0 & 0 & 1 & 0 & 4 & 6 & 4\\
        0 & 1 & 0 & 4 & 0 & 2 & 3 & 6\\
        0 & 0 & 1 & 3 & 0 & 6 & 0 & 4\\
        0 & 0 & 0 & 0 & 1 & 3 & 6 & 0\\
    \end{pmatrix}\quad\text{and}\quad
        G_2=\begin{pmatrix}
        1 & 0 & 0 & 0 & 6 & 3 & 4 & 0\\
        0 & 1 & 0 & 0 & 4 & 1 & 1 & 4\\
        0 & 0 & 1 & 0 & 1 & 1 & 4 & 6\\
        0 & 0 & 0 & 1 & 4 & 3 & 6 & 4\\
    \end{pmatrix}.
    \end{equation*}
We claim that there exist $\MDS$ codes $\Mcode_i^j$ such that $\dim(\Mcode_i^j)=j$ and $\Mcode_i^j\leq\code_i$ for $j=1,2,3$ and $i=1,2$.
In fact, for $j=1,2$ one can let $\Mcode_i^j$ be the code with generator matrix $M_i^j$, where
\begin{align*}
M_1^1&=\begin{pmatrix}
2 & 2 & 3 & 5 & 2 & 1 & 2 & 4\\
\end{pmatrix}\quad\text{and}\quad
M_2^1=\begin{pmatrix}
4 & 1 & 4 & 2 & 5 & 2 & 3 & 1\\
\end{pmatrix},\\
M_1^2&=\begin{pmatrix}
1 & 0 & 1 & 4 & 2 & 2 & 4 & 1\\
0 & 1 & 3 & 6 & 2 & 5 & 1 & 4
\end{pmatrix}\quad\text{and}\quad
M_2^2=\begin{pmatrix}
1&0&4&4&5&5&2&5\\
0&1&5&6&5&3&1&2
\end{pmatrix},
\end{align*}
Moreover, it can be checked that $\code_1\cap\code_2$ is an MDS code of dimension 3. This establishes the claim and shows that $\mu_j(\code_i)=j$ for $i=1,2$, $j=1,2,3$.
Since neither $\code_1$ nor $\code_2$ is $\MDS$, by~\Cref{lem:GMDSW_properties} we also have $\mu_4(\code_i^\perp)=5$ for $i=1,2$. Summarizing, we have $$\mu(\code_1)=\mu(\code_2)=(1,2,3,6).$$

Consider the dual codes. Each $\code_i^\perp$ contains $\MDS$ subcodes $\Ncode_i^j$ of dimension $j=1,2$, with generator matrices $N_i^j$ 
\begin{align*}
N_1^1&=\begin{pmatrix}
4 & 5 & 1 & 1 & 3 & 1 & 6 & 1\\
\end{pmatrix}\quad\text{and}\quad
N_2^1=\begin{pmatrix}
5 & 2 & 1 & 6 & 6 & 3 & 5 & 2\\
\end{pmatrix},\\
N_1^2&=\begin{pmatrix}
1& 0& 4& 1& 3& 4& 1& 1\\
0& 1& 3& 3& 6& 4& 4& 5
\end{pmatrix}\quad\text{and}\quad
N_2^2=\begin{pmatrix}
1& 0& 6& 2& 5& 3& 4& 2\\
0& 1& 2& 6& 3& 5& 4& 5
\end{pmatrix}.
\end{align*}     
It follows that $\mu_j(\code_i^\perp)=j$ for $j=1,2$, $i=1,2$.
We computed the values of $\mu_3(\code_i^\perp)$, $i=1,2$, by brute force search.

Let $L_3(\code_1^\perp)=\{D\leq\code_1^\perp\mid\dim(D)=3\}$, then it can be checked that
\begin{equation*}
\max\{d(D)\mid D\in L_3(\code_1^\perp\}=5
\end{equation*}
hence $\mu_3(\code_1^\perp)\neq3$.
By \Cref{lem:duality_GMDSW}, $\mu_4(\code_1)=6$. The same strategy applied to $\code_2$ yields $\max\{d(D)\;\:\;D\in L_3(\code_2^\perp\}=5$ and $\mu_4(\code_2)=6$.
Since $\mu_4(\code_i^\perp)=5$, by~\Cref{lem:GMDSW_properties} we have $\mu_3(\code_i^\perp)\in\{4,5\}$.
The code with generator matrix
\begin{equation*}
N_2^3=\begin{pmatrix}
1& 0& 0& 0& 6& 3& 4& 1\\
0& 1& 0& 0& 4& 1& 1& 4\\
0& 0& 1& 0& 1& 1& 4& 3\\
0& 0& 0& 1& 4& 3& 6& 6\\
\end{pmatrix}   
\end{equation*}
is a Reed-Solomon code, hence MDS. It can be checked by direct computation that it has a $3$-dimensional intersection with $\code_2^\perp$.
This implies $\mu_3(\code_2^\perp)=4$, hence $$\mu(\code_2^\perp)=(1,2,4,5).$$

We claim that $\mu_3(\code_1^\perp)\neq 4$. In fact, suppose by contradiction that $\mu_3(\code_1^\perp)=4$. Then there exists an $\MDS$ code $\Mcode$ of dimension 4 such that $\dim(\code_1^\perp\cap\Mcode)\geq3$.
Since $\code_1^\perp$ is not $\MDS$, then $\dim(\code_1^\perp\cap\Mcode)=3$, so $\code_1$ and $\Mcode$ share a $3-$dimensional subcode.
Since we are working over a prime field, $\Mcode$ is equivalent to an extended Reed-Solomon code by~\cite[Corollary 9.2]{ball2012sets}. Hence $\mu_3(\code_1^\perp)=4$ if and only if the lists $L_3(\code_1^\perp)$ and $L_3(\Mcode)=\{D\leq\Mcode\mid\dim(D)=3\}$ contain equivalent codes.
However, one can check that all codes in $L_3(\code_1^\perp)$ have a covering radius 4, while by the Supercode Lemma~\cite[Lemma 11.1.5]{huffman2010fundamentals} all codes in $L_3(\Mcode)$ have covering radius 5.
It follows that $\mu_3(\code_1^\perp)\neq 4$, hence $\mu(\code_1^\perp)=(1,2,5,5)$.  In particular, $\mu(\code_1^\perp)\neq\mu(\code_2^\perp)$.
\end{example}

\subsubsection{Recovering the duality of GMDS weights: $q\rightarrow\infty$}

As argued in the previous section, the lack of duality for GMDS weights is related to the lack of reciprocal containments between the codes in $\MDSFam$.
For example, the construction of~\Cref{ex:counterexample_duality_GMDSW} heavily relies on this feature, which is also highlighted in~\Cref{ex:MDS_bad_nestability}. As we observed in the previous section, this prevents the family of MDS codes from being a test family in general.
However, this does not preclude the existence of parameter regimes where the family of MDS codes is a test family.
This is the case when the field size $q$ is large enough compared to the length, as~\Cref{prop:MDS_nesting_q_big} shows.
Our statement's proof relies on the following variant of a result by Pellikaan.

\begin{proposition}[{\cite[Proposition 5.1]{pellikaan1996existence}}]\label{prop:pellikaan}
Let $\code\leq\F_q^n$ with $\Hd(\code^\perp)>t$, and $q>\max\{\binom{n}{i}\mid1\leq i\leq t\}$.
Then there exists a sequence $\Mcode_0\lneq\Mcode_1\lneq\ldots\lneq\Mcode_t$ of MDS codes contained in $\code$.
\end{proposition}

\begin{proof}
The statement can be proved by induction, following the argument of~\cite{pellikaan1996existence}. The only difference from~\cite{pellikaan1996existence} is that we suppose that the field $\F_q$ is large enough to satisfy the requirements that are met by taking a suitable extension in~\cite{pellikaan1996existence}.
\end{proof}

\begin{proposition}\label{prop:MDS_nesting_q_big}
Let $q\geq\binom{n}{\lfloor n/2\rfloor}$, then every MDS code $\Mcode\leq\F_q^n$ is contained in a fixed step chain of MDS codes with step $\ell=1$.
\end{proposition}

\begin{proof}
Let $\mathcal{M}$ be a $k$-dimensional MDS subcode of $\F_q^n$.
Since $\Hd(\Mcode)=n-k+1>\dim(\Mcode^\perp)$, by~\Cref{prop:pellikaan} we have a chain of MDS codes
$$0=\Mcode_0\lneq\Mcode_1\lneq\ldots\lneq\Mcode_{k-1}\lneq\Mcode_k=\Mcode.$$ Applying the same result to $\Mcode^\perp$, we obtain a chain
$$\Ncode_0\lneq\ldots\lneq\Ncode_{n-k}\lneq\Ncode_{n-k+1}=\Mcode^\perp.$$
Dualizing and combining the two chains yields the desired chain of MDS codes of consecutive dimensions 
$$0=\Mcode_0\lneq\Mcode_1\lneq\ldots\lneq\Mcode_{k-1}\lneq\Mcode_k=\Mcode
\lneq\Mcode_{k+1}\lneq\ldots\lneq\Mcode_{n-1}\lneq\Mcode_n=\F_q^n,$$
where $\Mcode_i=\Ncode_{n-i}^\perp$ for $0\leq i\leq n-k$.
Notice that~\Cref{prop:pellikaan} applies to MDS codes of every dimension, since we assumed $q\geq\binom{n}{\lfloor n/2\rfloor}\geq\max\{\binom{n}{i}\mid1\leq i\leq t\}$.
\end{proof}

\Cref{prop:MDS_nesting_q_big} shows that $\MDSFam$ is a metric test family for large enough $q$.
Hence, in this parameter regime, GMDS weights are invariants that satisfy the duality statement given by~\Cref{thm:duality}.

\begin{theorem}
Let $\code\subseteq\F_q^n$ be an $[n,k]_q$ code. If $q\geq \binom{n}{\lfloor n/2\rfloor}$, then
the sequences $\tau(\code^\perp)$ and $\tau(\code)$ determine each other via
\begin{equation*}
\{\tau_s(\code^\perp)\;:\; 1\leq s\leq n-k\}=[n]\setminus\{n+1-\tau_{r}(\code),\; 1\leq r\leq k\}.
\end{equation*}    
\end{theorem}

\subsubsection{Relation to subcode distances}\label{sec:GMDSW_subcode_dist}

Subcode distances are invariants of linear codes in metric vector spaces defined in~\cite{camps2025code}. We recall the definition in $(\F_q^n,\Hd)$. 

\begin{definition}\label{def:subcode_distances}
Let $\code\leq\F_q^n$ be a $k$-dimensional code. For $1\leq i\leq k$, the {\em $i^{th}$ subcode distance} is
\begin{equation*}
\alpha_i(\code)=\max\{\Hd(\Dcode)\mid\Dcode\leq\code,\dim(\Dcode)=i\}.
\end{equation*}
\end{definition}

The subcode distances are invariants of $\code$, but not of the matroid of $\code$, as shown in~\cite{camps2025code}. Moreover, they have a natural application as code distinguishers for inequivalent MDS codes. 
It turns out that, when the field size $q$ is large enough, the subcode distances and the GMDS weights are equivalent invariants.
This is due to the fact that, for $q$ large enough, every code can be embedded in a larger MDS code with the same minimum distance. This was proved in~\cite[Corollary 5.2]{pellikaan1996existence} over a field extension $\F_{q^N}$ of $\F_q$ of sufficiently large degree $N$.
We now give a different proof of this result.

\begin{lemma}\label{lem:MDS_conatinment_q_big}
Let $\code$ be an $[n,k,d]_q$ code. If $q\gg 0$, then there exists a $[n,n-d+1,d]_q$ code $\Mcode$ that contains $\code$.
\end{lemma}

\begin{proof}
If $\code$ is MDS, the statement is trivially valid. Hence, we assume $k<n-d+1$.
Consider the sets of codes
$$F=\{\Dcode\geqslant\code\;:\;\dim(\Dcode)=n-d+1\}\supseteq F^\prime=\{\Dcode\in F\;:\;d(\Dcode)<d\}.$$ A code $\Mcode$  as in the statement exists if $|F|-|F'|>0$.
It is well known that \smash{$|F|=\binom{n-k}{d-1}_q$}, therefore our statement is proved if \smash{$1-\frac{|F'|}{|F|}>0$}.
We know from~\cite[Theorem 5.1]{byrne2020partition} that 
\begin{equation*}
\frac{|F'|}{|F|}\leq\frac{(q^{n-d+1}-q^k)|B^*_{d-1}|}{(q-1)(q^n-q^k)}\;,
\end{equation*}
where $B_{d-1}^*=\{x\in\F_q^n\;:\;x\neq0,\Hw(x)\leq d-1\}$.
Then $\Mcode$ exists provided that
\begin{equation*}
\sum_{i=1}^{d-1} {n\choose i}(q-1)^i=|B^*_{d-1}|<\frac{(q-1)(q^n-q^k)}{q^{n-d+1}-q^k},
\end{equation*}
and this is the case for any $q\gg0$, as the RHS of the inequality is $\geq C_1q^d$, while $|B^*_{d-1}|\leq C_2q^{d-1}$, where $C_1$ and $C_2$ are positive constants.
\end{proof}

The next result establishes a relation between subcode distances and GMDS weights.

\begin{proposition}\label{prop:GMDSW_code_distances}
Assume that $q\gg0$ and let $\code\leq\F_q^n$ be a $k$-dimensional code. Then
\begin{equation*}
\mu_r(\code)=n+1-\alpha_r(\code)
\end{equation*}
for $1\leq r\leq k$.
\end{proposition}

\begin{proof}
The definition of GMDS weights can be rewritten as
$$\mu_r(\code)=\min\{\dim(\Mcode)\;:\;\Mcode\supseteq \Dcode,\Dcode\leq\code,\dim(\Dcode)=r\}.$$
Assume that $q$ is large enough so that~\Cref{lem:MDS_conatinment_q_big} holds and let $\Dcode\leq\code$ be a subcode with minimum distance $d$. By~\Cref{lem:MDS_conatinment_q_big}, $\Dcode$ is contained in an MDS code $\Mcode$ with minimum distance $d$ and dimension $n-d+1$.
Applying this reasoning to all subcodes of $\code$, we have
    \begin{align*}
        \mu_r(\code)&=\min\{\dim(\Mcode)\;:\;\Mcode\supseteq\Dcode, \Dcode\leq\code, \dim(\Dcode)=r\}\\
        &=\min\{n+1-\Hd(\Dcode)\;:\;\Dcode\leq\code,\dim(\Dcode)=r\}\\
        &=n+1-\max\{\Hd(\Dcode)\;:\;\Dcode\leq\code,\dim(\Dcode)=r\}\\
    \end{align*}
for $1\leq r\leq k$.   
\end{proof}

\subsection{Generalized MRD weights}

In this section, we discuss generalized MRD weights, which are the analog of GMDS weights in the rank metric space $\F_q^{n\times m}$.

\begin{definition}
Let $\code\leq\F_q^{n\times m}$ be a $k$-dimensional code. For $r=1,\ldots,k$ the {\em $r^{th}$ generalized $\MRD$ (GMRD) weight} of $\code$ is
\begin{align*}
\mu_{r}(\code)&=\min\{\dim(\Mcode)\;:\;\Mcode\text{ is an $\MRD$ code},\, \dim(\Mcode\cap\code)\geq r\}.
\end{align*}
The sequence of GMRD weights is denoted by $\mu(\code)$.
\end{definition}

The following lemma collects some straightforward properties of GMRD weights.

\begin{lemma}\label{lem:properties_GMRDW}
Let $\code\leq\F_q^{n\times m}$ be a $k$-dimensional code, $k^\perp=nm-k$. Then:
\begin{enumerate}
\item the GMRD weights are invariants of $\code$,
\item $\lceil \frac{r}{m}\rceil m \leq\mu_r(\code)\leq \lceil\frac{nm-k+r}{m}\rceil m$ for $r=1,\ldots,k$,
\item $\mu_1(\code)\leq\mu_2(\code)\leq\ldots\leq\mu_k(\code)$,
\item $\mu_{k^\perp+r-\mu_r(\code)}(\code^\perp)\leq nm-\mu_r(\code)$ for $r=1,\ldots,k$.
\end{enumerate}
\end{lemma}

\begin{proof}
The arguments for 1., 3. and 4. are the same as in~\Cref{lem:GMDSW_properties}, replacing $\MDS$ codes with $\MRD$ codes.
The proof of 2. follows the same steps as the bounds on the $T$-weights in~\Cref{lem:T-weights_basic_prop}.
\end{proof}

As in~\Cref{sec:rank}, we derive strict inequalities between weights whose indices differ by $m$ from the property that the elements of the corresponding test families are nestable.
$\MRD$ codes are not necessarily nestable, i.e., there are $\MRD$ codes of which are not contained in an $\MRD$ code whose dimension is larger by $m$. Dually, we can find $\MRD$ codes that do not contain an $\MRD$ subcode of codimension $m$.
An example of such codes over the field $\F_2$ is given in~\cite[Example 38]{byrne2017covering}.
This is related to the study of the covering radius and maximality degree of a code.

The next result is the analog of~\Cref{lem:duality_GMDSW} in the rank metric, and it is proved in a similar fashion.

\begin{proposition}\label{lem:duality_GMRDW}
Let $\code\subseteq\F_q^{n\times m}$ be a $k-$dimensional code, $n\leq m$. We have
\begin{equation*}
\mu_t(\code)\in[nm]_m\setminus\{(n+1)m-\mu_r(\code^\perp)\;:\;r=_mt-k\}
\end{equation*}
for every $t=1,\ldots,k$.
\end{proposition}

\begin{proof}
Let $k^\perp=nm-k$ and suppose by contradiction that there exist $t\in[k]$ and $r\in[k^\perp]$ such that $\mu_t(\code)+\mu_r(\code^\perp)=(n+1)m$. 
By~\Cref{lem:properties_GMRDW}.4 
\begin{equation*}
\mu_t(\code)+\mu_{k^\perp+t-\mu_t(\code)}(\code^\perp)\leq nm,
\end{equation*} hence $r>k^\perp+t-\mu_t(\code)$ by~\Cref{lem:properties_GMRDW}.3. Since $m\mid r-t+k$, this is equivalent to $r\geq k^\perp+t+m-\mu_t(\code)$. 
By~\Cref{lem:properties_GMRDW}.4 one also has
\begin{equation*} 
\mu_r(\code^\perp)+\mu_{k+r-\mu_r(\code^\perp)}(\code)\leq nm,
\end{equation*} 
which by~\Cref{lem:properties_GMRDW}.3 implies that $t>k+r-\mu_r(\code^\perp)$. Since $m\mid r-t+k$, this is equivalent to $t\geq k+r+m-\mu_r(\code^\perp)$. 
Summing the two inequalities yields $t+r\geq nm+t+r+2m-\mu_t(\code)-\mu_r(\code^\perp)$, i.e., $\mu_t(\code)+\mu_r(\code^\perp)\geq (n+2)m$, a contradiction.
\end{proof}

\begin{corollary}
For all $1\leq r\leq m$ the sets
\begin{equation*}
\mu^r(\code)=\{\mu_t(\code)\;:\;t=r \mod m\}\quad\text{and}\quad\mu^{r-k}(\code^\perp)=\{\mu_t(\code^\perp)\;:\;t= r-k\mod m\}
\end{equation*}
are disjoint subsets of $[nm]_m$.
\end{corollary}

Similarly to the case of GMDS weights, the difference between~\Cref{lem:duality_GMRDW} and a duality statement is that the subsequences of weights that we consider are not, in general, complementary subsets of $[nm]_m$.
The next example shows that there exist codes with the same GMRD weights, whose dual codes have different GMDR weights.

\begin{example}
For $q=2$, $n=m=4$, a classification of $\MRD$ codes up to equivalence can be found in~\cite{sheekey201913}. There are 3 equivalence classes of codes with minimum distance 4 (from the classification of semifields in~\cite{knuth1963finite,aman2013enumerating}) and one class with minimum distance 3, (the class of Gabidulin codes).
Taking duals, we have one class with minimum distance 2 and three classes with minimum distance 1.
Of the three equivalence classes of dimension 4, two have covering radius~3 and one has covering radius~2.
In particular, one of the two classes with covering radius 3 contains a Gabidulin code $\code_1$. Let $\code_2$ be a representative of the class with covering radius 2.
Since both $\code_1$ and $\code_2$ are $\MRD$ codes of dimension 1, we have $\mu(\code_1)=\mu(\code_2)=(1,1,1,1)$.
However, while $\code_1$ is contained in a Gabidulin code of dimension 8, $\code_2$ is not contained in any MRD code of dimension 8, as this would contradict the Supercode Lemma~\cite[Lemma 11.1.5]{huffman2010fundamentals}.
It follows that $\code_1^\perp$ contains an $\MRD$ subcode of dimension 8, while $\code_2^\perp$ does not contain any $\MRD$ subcode of this dimension.
Therefore, $\mu_8(\code_1^\perp)=2$ and $\mu_8(\code_2^\perp)>2$. As $\code_2^\perp$ is $\MRD$, it has $\mu_{12}(\code_2^\perp)=3$, and hence $\mu_8(\code_2^\perp)=3$. This shows that 
\begin{equation*}
\{\mu_t(\code_1^\perp)\mid t=_40\}\neq\{\mu_t(\code_2^\perp)\mid t=_40\},
\end{equation*}
showing that no duality statement similar to~\Cref{thm:duality} holds. 
\end{example}

It is natural to wonder whether the GMRS weights are related to subcode distances in the rank metric.
The definition of subcode distances~\cite{camps2025code} in the rank metric space $(\F_q^{n\times m},\rkd)$ is similar to~\Cref{def:subcode_distances}.

\begin{definition}
Let $\code\leq\F_q^{n\times m}$ be a $k$-dimensional code. For $1\leq i\leq k$, the {\em $i^{th}$ subcode distance} is
\begin{equation*}
\alpha_i(\code)=\max\{\rkd(\Dcode)\mid\Dcode\leq\code,\, \dim(\Dcode)=i\}.
\end{equation*}
\end{definition}

It is easy to show that the first subcode distance is the maximum weight of the code and is related to the first GMRD weight as follows.

\begin{lemma}
For every code $\code\leq\F_q^{n\times m}$, $n\leq m$, we have $$\mu_1(\code)=n+1-\maxrk(\code).$$
\end{lemma}

\begin{proof}
Let $c\in\code$ be an element of maximum rank. There is a $\MRD$ code $\Mcode$ whose minimum distance is equal to $\rank(c)$, and two linear transformations $N\in \GL_n(\F_q)$, $M\in\GL_m(\F_q)$ such that $c=NxM$ for some $x\in\Mcode$. It suffices, e.g., to let $\Mcode$ be the Gabidulin code generated by $c$ and its first $n-\rank(c)$ Frobenius powers.
It follows that $N\Mcode M$ is an $\MRD$ code of dimension $m(n+1-\rank(c))$ and with $\dim(\code\cap N\Mcode M)\geq1$. The thesis follows since this is the least possible dimension for an MRD code that contains an element of $\code$.
\end{proof}

The key argument to establish the asymptotic connection between the GMDS weights and the subcode distances in~\Cref{sec:GMDSW_subcode_dist} is that, for $q$ large enough, every code of minimum distance $d$ is contained in an MDS code with the same minimum distance.
The proof is closely related 
to the one that establishes the density of MDS codes, among codes of the same dimension, given in~\cite{byrne2020partition}.
Studying the density of MRD codes is a more complex problem. In fact, unlike MDS codes,
MRD codes are known to be sparse within codes sharing a given dimension, with very sporadic exceptions~\cite{gruica2022common}. In particular, the proof technique applied 
in~\Cref{sec:GMDSW_subcode_dist} does not naturally translate from the Hamming to the rank metric. 

However, if one considers subcode distances in the vector rank-metric space $(\F_{q^m}^n,\rkd)$, then one can establish a relation between GMRD weights and subcode distances similar to that of \Cref{sec:GMDSW_subcode_dist}. In fact, in this context the density of MRD codes is well-understood and arguments similar to those of \Cref{sec:GMDSW_subcode_dist} can be made.

\end{sloppypar}

\bibliographystyle{plain}
\bibliography{refs}

\end{document}